\documentclass[12pt,reqno]{amsart}
\usepackage{amsmath,amssymb,amsthm}
\usepackage{latexsym}
\usepackage{amsfonts}
\usepackage{bbm,bm,dsfont} 
\usepackage{color} 
\usepackage{graphicx}
\usepackage{hyperref}

\theoremstyle{definition}
\newtheorem{proposition}{Proposition}

\newtheorem{example}{Example}
\newtheorem{remark}{Remark}
\newtheorem{theorem}{Theorem}

\newtheorem{corollary}{Corollary}

\newcommand{\cal}{\mathcal}

\newcommand{\hM}{\mathcal{I}} 

\newcommand{\hT}{\mathcal{T}}

\newcommand{\N}{\mathbb N} 
\newcommand{\C}{\mathbb C} 
\newcommand{\T}{\mathbb T} 

\newcommand{\fii}{\varphi} 
\newcommand{\Om}{\Omega} 

\newcommand{\hi}{\mathcal{H}} 
\newcommand{\ki}{\mathcal{K}} 
\newcommand{\li}{\mathcal{L}} 

\newcommand{\id}{\mathbbm{1}} 

\newcommand{\lh}{\mathcal{L(H)}} 
\newcommand{\kh}{\mathcal{L(K)}} 
\renewcommand{\th}{\mathcal{T(H)}} 
\newcommand{\sh}{\mathcal{S(H)}} 
\newcommand{\sk}{\mathcal{S(K)}} 

\newcommand{\tr}[1]{\mathrm{tr}\left[#1\right]} 
\newcommand{\ptr}[1]{\mathrm{tr}_\mathcal{K}\left[#1\right]} 
\def\<{\langle} 
\def\>{\rangle} 
\newcommand{\kb}[2]{|#1 \rangle\langle #2|} 
\newcommand{\ip}[2]{\left\langle #1 | #2 \right\rangle} 
\newcommand{\rnk}[1]{\text{rank}(#1)} 

\newcommand{\Ao}{\mathsf{A}} 
\newcommand{\Po}{\mathsf{P}} 
\newcommand{\Mo}{\mathsf{M}} 
\newcommand{\Zo}{\mathsf{Z}} 

\def\d{{\mathrm d}} 

\newcommand{\hd}{\hi_\oplus}

\newcommand{\mc}[1]{\mathcal{#1}}

\begin{document}

\title[Minimal normal measurement models of quantum instruments]{Minimal normal measurement models of quantum instruments}

\author{Juha-Pekka Pellonp\"a\"a}
\email{juhpello@utu.fi}
\address{Turku Centre for Quantum Physics, Department of Physics and Astronomy, University of Turku, FI-20014 Turku, Finland}

\author{Mikko Tukiainen}
\email{mjatuk@utu.fi}
\address{Turku Centre for Quantum Physics, Department of Physics and Astronomy, University of Turku, FI-20014 Turku, Finland}

\begin{abstract}
In this work, we study the minimal normal measurement models of quantum instruments. 
We show that usually the apparatus' Hilbert space in such a model is unitarily isomorphic to the minimal Stinespring dilation space of the instrument.
However, if the Hilbert space of the system is infinite dimensional and the multiplicities of the outcomes of the associated observable (POVM) are all infinite then this may not be the case.
In these pathological cases, the minimal apparatus' Hilbert space is shown to be unitarily isomorphic to the instrument's minimal dilation space augmented by one extra dimension.
 We also point out errors in earlier papers of one of the authors (J-P.P.).
\newline

\end{abstract}

\maketitle

\section{Introduction}
Quantum measurement theory provides an operational connection between mathematical Hilbert space operator theory and actual physically realizable measurements. The theory builds upon the natural fact that every measurement realizing a \textit{quantum device}, \textit{i.e.}, a quantum observable, channel or instrument, can be described as a protocol in which an observed system is brought in contact with a measurement apparatus and after interaction the value of the measured quantity is read from the apparatus' pointer scale \cite{QTM96}. Since Physics is an experimental science, knowing the limitations of the measurement theory is of paramount importance.

One of the fundamental results in quantum measurement theory is that every (completely positive) quantum instrument can be realized in a \textit{normal measurement} of the associate observable \cite{Ozawa84}. These measurements are important 
 since ideal experimental setups can be described by such models. In this work we solve the \textit{minimal} normal measurement models associated to a given quantum instrument or observable. Such models exclude all unnecessary degrees of freedom and may, for example, lead to optimal control over detrimental effects caused by inevitable noise. 

As a tool to approach the problem of finding out the minimal normal measurement models, we use \textit{Stinespring dilations}, that are isometric expansions of the given devices into their purifications. This approach allows us to re-phrase the above problem into a question `when does an operator $U: \hi_A \otimes \hi_B \rightarrow  \hi_A \otimes \hi_B$ defined via $U(\psi \otimes \xi) = Y \psi$, where $Y: \hi_A \rightarrow  \hi_A \otimes \hi_B$ is a given linear isometry and $\xi \in \hi_B$ is some unit vector, extend to a unitary operator on $\hi_A \otimes \hi_B$?'' Mathematically, this \textit{unitary extension problem} can be related to well known and established results on extendability of partial isometries \cite{AkhiGlaz93} and unitary dilations of contractive maps \cite{Halmos82}.

In addition, the work presented here has also applications in problems associated with optimizing the `Church of a Larger Hilbert Space'' in sense of minimality. For instance, in the field of open quantum systems our results may be used to solve a minimal environment in which a composite system evolution is unitary, hence having a minimal number of (potentially harmful) environmental degrees of freedom.

Our study is organized as follows. In the preliminary Sect.\,\ref{sec:prel} we introduce the mathematical concepts and notations needed throughout this paper. In Sect.\,\ref{sec:isometry} we study the problem of extending linear isometries to unitary mappings. Finally, in the last two Sects.\,\ref{sec:discrete} and \ref{sec:continuous} we present our main results by solving the minimal normal measurement models of quantum devices. 

We wish to emphasize that most of the work presented here has already been carried out in Refs.\,\cite{Pel2013II, Pel2014, PeENT}. Unfortunately, these papers contain some errors related to the aforementioned unitary extension problem -- in some very pathological cases the extension is not possible; see \hyperref[errata]{Errata} in the end of this paper. We 
show that 
this problem can be resolved by adding one extra dimension to the ancillary space $\hi_B$.

\section{Preliminaries}\label{sec:prel}

Throughout this article, we let $\hi$ be a complex Hilbert space and denote the set of \textit{bounded operators}\footnote{Similarly, $\mc L(E)$ denotes the set of bounded linear maps on a Banach space $E$ which is either $\th$ or $\lh$ in this article.} on $\hi$ by $\lh$, the set of \textit{projections} by $\mc P(\hi)$ and the set of \textit{trace class operators} by $\th$. Furthermore, we let $\sh = \{ \rho \in \th \, | \, \rho \geq 0, \, \tr{\rho}=1\}$ denote the convex set of \textit{quantum states}, and we call its extremal elements as \textit{pure states}. It is well known, that  pure states are of the form $\rho = | \xi \rangle \langle \xi |\in \mc P(\hi)$ for some unit vector $\xi \in \hi$.

\subsection{Observables} 
Let $\Omega$ be a set and $\Sigma\subseteq 2^\Omega$ a $\sigma$-algebra. \textit{Quantum observables} are identified with normalized positive operator valued measures (POVMs), that is, (weakly) $\sigma$-additive mappings $\Mo: \Sigma \rightarrow \lh$ such that $\Mo(X)\ge 0$ for all $X\in\Sigma$ and
$\Mo(\Omega) = \id_\hi$, where $\id_\hi$ stands for the identity operator on $\hi$. If $\Om_N=\{x_1,\,x_2,...\}$ with $N\leq\infty$ elements and $\Sigma_N = 2^{\Omega_N}$ is the corresponding $\sigma$-algebra of a POVM $\Mo$, then $\Mo$ 
can be viewed as a collection $(\Mo_1,\,\Mo_2,\ldots )$ of positive  operators $\Mo_i:=\Mo(\{x_i\})\in\lh$, such that $\sum_{i=1}^N\Mo_i=\id_\hi$ (weakly). In this case, we say that $\Mo$ is a \textit{discrete $N$-outcome observable}. By reducing the outcome space, we may (and will) assume that $\Mo_i\ne 0$ for all $i<N+1$.

A special class of observables consists of POVMs $\Po:\Sigma \rightarrow \lh$ whose range contain only projections, \textit{i.e.}, $\Po(X) \in \mc P (\hi)$ for every $X \in \Sigma$. We call them normalized projection valued measures (PVMs) or \textit{sharp observables}. 

Any observable $\Mo:\Sigma \rightarrow \lh$ has a {\it Naimark dilation} $(\ki, \Po, J)$ into a sharp one, that is, there exist an auxiliary Hilbert space $\ki$, a PVM $\Po:\Sigma \rightarrow \mc L (\ki)$ and a linear isometry $J: \hi \rightarrow \ki$ such that $\Mo(X) = J^* \Po(X) J$ for all $X\in\Sigma$. A dilation $(\ki, \Po, J)$ is called \textit{minimal} if the linear span of vectors $\Po(X) J\psi$, $X\in\Sigma$, $\psi\in\hi$, is dense in $\ki$. It is well known that in this case $\Mo$ is sharp if and only if $J$ is unitary. 

\subsection{Measurement models} 
A quantum measurement is mathematically modelled by a 4-tuple $(\ki, \Zo, \mc V, \eta)$ where $\ki$ is the \textit{apparatus' Hilbert space}, $\Zo: \Sigma \rightarrow \kh$ is the \textit{pointer observable} (POVM), $\mc V: \mc T( \hi \otimes \ki) \rightarrow \mc T ( \hi \otimes \ki)$ is a completely positive trace-preserving (CPTP) linear map describing the {\it interaction} between the system and the apparatus, and $\eta \in \sk$ is the {\it initial state of the probe}. Measurement model $(\ki, \Zo, \mc V, \eta)$ actualizes the measured observable $\Mo$ via the {\it probability reproducibility condition}
\begin{eqnarray} \label{eq:repro}
\tr{\Mo(X) \rho} &=& \tr{ \id_\hi \otimes \Zo(X) \, \mc V(\rho \otimes \eta) }, \quad X \in \Sigma,\; \rho \in \sh.
\end{eqnarray} 
For this reason we call $(\ki, \Zo, \mc V, \eta)$ a measurement model for $\Mo$ or shortly an $\Mo$-\textit{measurement}.
\cite{QTM96}

Every observable has an infinite number of different measurements. Among these there are \textit{normal measurements} in which the pointer $\Po:\Sigma \rightarrow \kh$ is a sharp observable, the interaction is of the form $\rho\otimes\xi \mapsto U \rho \otimes \xi U^*$ for some unitary operator $U$ on $\hi\otimes \ki$ and $\eta = |\xi\rangle\langle\xi|$ for some unit vector $\xi\in \ki$ \cite{Ozawa84} -- in such a case we write shortly $(\ki, \Po, U, \xi)$.
It can be easily derived from Eq.\,(\ref{eq:repro}), that the measured observable induced by a normal model $(\ki, \Po, U, \xi )$ is given by 
\begin{eqnarray}\label{obs:dilation}
\Mo(X) = Y_\xi^*\,  U^*(\id_\hi \otimes \Po(X))U \, Y_\xi \, , \qquad X\in \Sigma, 
\end{eqnarray} 
where $Y_\xi: \hi \rightarrow \hi \otimes \ki$ is a linear operator defined via $Y_\xi(\varphi)=\varphi\otimes \xi$ for all $\varphi\in\hi$. Since $Y_\xi^* Y_\xi = \id_\hi$, $Y_\xi$ is an isometry and Eq.\,(\ref{obs:dilation}) constitutes a Naimark dilation of $\Mo$. Note that this dilation need not to be minimal. For instance, if $\Mo$ is a sharp observable, then Eq.\,\eqref{obs:dilation} is never minimal Naimark dilation \cite{LaYl04}.

\subsection{Instruments and channels} 
The features of a measurement model $(\ki, \Zo, \mc V, \eta )$, that refer to the measured system, are neatly captured by the associated (Schr\"{o}dinger) \textit{instrument}: a completely positive trace non-increasing mapping $\hM : \Sigma\to\li\big(\th\big)$ defined via
\begin{eqnarray}
\hM(X)(T) = \ptr{ \id_\hi \otimes \Zo(X) \,  \mc V( T \otimes \eta)}, \qquad X \in \Sigma, \; T \in \th,
\end{eqnarray} 
where $\text{tr}_\ki$ stands for the partial trace over $\ki$.
For example, the measurement statistics of the measured observable $\Mo$ are given by 
\begin{eqnarray}\label{prel:repro}
\tr{\Mo(X) \,  \rho} = \tr{\hM(X)(\rho)}, \qquad \rho \in \sh.
\end{eqnarray} 
We say that the observable $\Mo$ in Eq.\,(\ref{prel:repro}) is \textit{associated} to the instrument $\hM$. The mapping $\hM$ can be equivalently viewed in its dual form $\hM^* : \Sigma \to\li\big(\lh\big)$ defined via 
\begin{eqnarray}
\tr{\hM^*(X)(B) \, T} = \tr{B \, \hM(X)(T)},  \qquad B \in \lh, \; T\in\th.
\end{eqnarray} 
Throughout this article, we will be mainly working in this \textit{Heisenberg picture}. Note that $\Mo(X)=\hM^*(X)(\id_\hi)$ for all $X\in\Sigma$.

As a generalization of the previously mentioned Naimark dilation theorem, any instrument $\hM : \Sigma\to\li\big(\th\big)$ has a \textit{Stinespring dilation} $(\ki, \Po, Y)$, \textit{i.e.}, there exists a Hilbert space $\ki$, a PVM $\Po:\Sigma \rightarrow \mc L (\ki)$ and a linear isometry $Y: \hi \rightarrow \hi\otimes \ki$ such that
\begin{eqnarray}
\hM^*(X)(B) &=& Y^* \left(  B \otimes \Po(X)\right) Y, \qquad B\in\lh.
\end{eqnarray} A Stinespring dilation $(\ki, \Po, Y)$ for $\hM$ is called \textit{minimal} if the vectors $\left( B \otimes \Po(X)\right) Y  \psi$, $B \in \lh$, $X \in \Sigma$, $\psi \in \hi$, span a dense subspace of $\hi \otimes \ki$. 
The instrument $\hM$ corresponding to a normal measurement model $(\ki, \Po, U, \xi)$ has a particularly simple form, namely,
\begin{eqnarray}\label{prel:instrudilation}
\hM^*(X)(B) = Y_\xi^* U^* ( B \otimes \Po(X) ) U Y_\xi, \qquad  X \in \Sigma,  \; B\in \lh.
\end{eqnarray}
Obviously Eq.\,(\ref{prel:instrudilation}) constitutes a Stinespring dilation for the instrument $\hM$ and it has been proven in Ref.\ \cite{Ozawa84} that every instrument has a measurement dilation of such form. Determining when Eq.\,(\ref{prel:instrudilation}) defines a minimal Stinespring dilation of $\hM$ is one of the main purposes of this article.


Any instrument $\hM: \Sigma  \rightarrow \li\big(\th\big)$ {\it induces} a (Schr\"odinger) {\it channel} (a CPTP linear map) $\mc E:\mc T (\hi) \rightarrow \mc T (\hi)$ describing the total state transformation induced by the corresponding measurement $( \ki, \Zo, \mc V, \eta )$ via 
\begin{eqnarray}\label{prel:inducing}
\mc E(T) = \hM(\Omega)(T) = \ptr{\mc V(T \otimes \eta)}, \qquad T \in \th.
\end{eqnarray}
Conversely, any channel can be induced in this way from some measurement $(\ki, \Zo, \mc V, \eta)$.
 Eq.\,(\ref{prel:inducing}) can be equivalently expressed in the Heisenberg picture as
\begin{eqnarray}
\mc E^*(B) = \hM^*(\Omega)(B),
\qquad B \in \lh.
\end{eqnarray} Note that the ``trace-preserving''-property in the Schr\"odinger picture corresponds to \textit{unitality}, $\mc E^*(\id_\hi) = \id_\hi$, in the Heisenberg picture. 

A Stinespring dilation of a channel is naturally inherited from its inducing instrument which, however, is not necessarily minimal. Furthermore, any completely positive trace non-increasing linear map $\mc C: \mc T (\hi) \rightarrow \mc T (\hi)$ has a \textit{Kraus representation}, that is there exists a sequence of \textit{Kraus operators} $A_1, \, A_2,...  \in \lh$, with $r \leq \infty$ elements, satisfying $\sum_{s=1}^r A_s^* A_s \leq \id_\hi$ such that
\begin{eqnarray}
\mc C(T) &=& \sum_{s=1}^r A_s T A_s^*, \qquad T\in \th,
\end{eqnarray}
which equivalently in the Heisenberg picture reads
\begin{eqnarray}
\mc C^*(B) &=& \sum_{s=1}^r A_s^* B A_s, \qquad B\in \lh.
\end{eqnarray}

Obviously, the Kraus representation is not unique. We say that the representation of $\mc C$ having the smallest number of non-zero Kraus operators is \textit{minimal}, in which case the corresponding minimal number $r$ is called as the \textit{(Kraus) rank} of $\mc C$.
Note that $\mc C$ is a quantum channel if and only if $\sum_{s=1}^r  A_s^* A_s = \id_\hi$.

A channel $\mc E$ and an observable $\Mo$ are said to be \textit{compatible}, if there exists an instrument having $\mc E$ and $\Mo$ as its induced channel and associated observable, respectively. It is well known that not every channel and observable are compatible. We end this section with an example that illustrates this phenomenon.

\begin{example} 
Let $\hM: \Sigma \rightarrow \li\big(\th\big)$ be an instrument. If the induced channel $\mc E$ has Kraus rank 1, that is, $\mc E(T) = \mc I(\Omega)(T) = A T A^*$, for some linear isometry $A:\hi \rightarrow \hi$ and for all $T\in \th$, then the associated observable $\Mo:\Sigma \rightarrow \lh$ is necessarily \textit{trivial}, \textit{i.e.}, $\Mo(X) = p(X) \, \id_\hi$, $X\in\Sigma$, for some probability measure $p:\Sigma \rightarrow [0,1]$. In other words, isometric channels are only compatible with trivial observables. This is a straightforward generalization of the well known ``no information without disturbance''-result; see the proof given in \cite[p.\ 32]{QTM96} with a substitution $\mc I(X)(T) \leftrightarrow A^* \mc I(X)(T) A$. On the contrary, it is easy to verify that
each trivial observable is compatible with every quantum channel.
\end{example}

\section{Unitary extension problem}\label{sec:isometry}

Let $\hi_A$ and $\hi_B$ be two (possibly infinite dimensional non-trivial) Hilbert spaces and let $\id_A$ and $\id_B$ denote their identity operators, respectively.
Denote their \textit{tensor product} Hilbert space by $\hi_{AB}:=\hi_A\otimes\hi_B$ and let $\id_{AB}:=\id_A\otimes\id_B$. Clearly $\dim(\hi_{AB})= \dim(\hi_A) \dim(\hi_B)$. We assume that $\hi_A$ is separable and thus has a countable orthonormal basis $\{h_n\}_{n=1}^{\dim(\hi_A)}$. We do not, however, assume the separability of $\hi_B$ in this section unless otherwise stated.

Let $Y:\,\hi_A\to\hi_{AB}$ be any linear isometry, that is, it is of the form
\begin{eqnarray}
Y = Y \sum_{n=1}^ {\dim(\hi_A)} \kb{h_n}{h_n} = \sum_{n=1}^ {\dim(\hi_A)} \kb{y_n}{h_n},
\end{eqnarray}
where the vectors $y_n:=Yh_n\in\hi_{AB}$ form an orthonormal set. Indeed, from $Y^* Y=\id_A$ one gets $\ip{y_n}{y_m} = \ip{h_n}{Y^* Y h_m} = \delta_{nm}$ (the Kronecker delta).
Since 
\begin{eqnarray}
P := Y Y^*= \sum_{n=1}^ {\dim(\hi_A)} \kb{y_n}{y_n}
\end{eqnarray} 
is a projection, we may use the Hilbert projection theorem to write $\hi_{AB}=P(\hi_{AB}) \oplus P^\perp (\hi_{AB})$ where $P^\perp := \id_{AB} -P$. Here the subspace $P(\hi_{AB})$ corresponds to the image of the isometry $Y$, \textit{i.e.}, 
$$ 
P(\hi_{AB})=Y(\hi_A)=\overline{{\rm span}\{y_n\,|\,1\le n<\dim(\hi_A)+1\} }.
$$ 
Since $\dim(P(\hi_A)) =\dim(Y(\hi_A)) = \dim (\hi_A)$ we get
\begin{eqnarray}\label{dimensions}
\dim(\hi_A) \dim(\hi_B) &=&  \dim (P(\hi_{AB})) +  \dim (P^\perp(\hi_{AB})) \nonumber \\
&=& \dim(\hi_A)+\dim (P^\perp(\hi_{AB})).
\end{eqnarray}

The main question of this section is: When does there exist a {\it unitary operator}\footnote{Note that $U:\hi_{AB}\rightarrow\hi_{AB}$ can always be extended to a partial isometry, but there may exist non-zero vectors $\varphi\in\hi_{AB}$ such that $U\varphi=0$.
} $U:\,\hi_{AB}\to\hi_{AB}$ and a unit vector $\xi\in\hi_B$, such that 
\begin{eqnarray}\label{isom:extension}
U(\psi\otimes\xi)=Y\psi,
\end{eqnarray} for all $\psi\in\hi_A\,$? We will call this the \textit{unitary extension problem.} We immediately notice that this problem reduces to finding only the unitary operator, since the unit vector $\xi$ may be fixed arbitrarily. Indeed, for any unit vector $\xi' \in \hi_B$, we have $\xi' = U_B\, \xi$ for some unitary operator $U_B$ on $\hi_B$ and thus $$U(\psi \otimes \xi) = U (\id_{A} \otimes U_B)^* (\psi \otimes U_B\xi) = U' (\psi \otimes \xi')$$ where $U' := U (\id_{A} \otimes U_B)^*$ is a unitary operator on $\hi_{AB}$. 

By writing
\begin{eqnarray}\label{jdfjhfdhj}
Y(\hi_A) \oplus P^\perp(\hi_{AB})=\hi_{AB}=[\hi_A\otimes(\C\xi)]\oplus[\hi_A\otimes(\C\xi)]^\perp
\end{eqnarray}
and noting that the (separable) subspaces $Y(\hi_A)$ and $\hi_A\otimes(\C\xi)$ are unitarily isomorphic\footnote{Since their orthonormal bases 
$\{y_n\}_{n=1}^{\dim(\hi_A)}$ and $\{h_n\otimes\xi\}_{n=1}^{\dim(\hi_A)}$ are of the same cardinality.}
one sees that such a $U$ exists (for a given $\xi$) if and only if the subspaces $P^\perp(\hi_{AB})$ and $[\hi_A\otimes(\C\xi)]^\perp$ are unitarily isomorphic. This happens at least when
\begin{itemize}

\item[(i)] $\hi_B$ is not separable.
Indeed, then $\hi_{AB}$ is not separable and, for any unit vector $\xi\in \hi_B$, $P^\perp(\hi_{AB})$ and $[\hi_A\otimes(\C\xi)]^\perp$ are non-separable (infinite dimensional) subspaces having orthonormal bases of the same cardinality, \textit{i.e.}, these subspaces are unitarily isomorphic.
\item[(ii)] $\hi_A$ is finite dimensional. Since $\dim (\hi_A) < \infty$, we can solve $\dim (P^\perp(\hi_{AB}))$ from Eq.\,(\ref{dimensions}):\footnote{If $\dim(\hi_B)=\infty$ then $\dim (P^\perp(\hi_{AB}))=\infty-\dim(\hi_A)=\infty$.} 
$$
\dim (P^\perp(\hi_{AB})) =\dim(\hi_A) \dim(\hi_B)-\dim(\hi_A).
$$
By fixing any unit vector $\xi\in\hi_B$ and an orthonormal basis $\mc L$ of $\hi_{AB}$ such that $h_n\otimes\xi\in \mc L$, for all $n\le \dim(\hi_A)$, one sees that the dimension of $[\hi_A\otimes(\C\xi)]^\perp$ is $\dim([\hi_A\otimes(\C\xi)]^\perp) =\dim(\hi_A) \dim(\hi_B)-\dim(\hi_A) = \dim(P^\perp(\hi_{AB}))$.
\end{itemize}

The above shows that problems in extending an operator $U$ defined in Eq.\,(\ref{isom:extension}) into a unitary one may only occur when 
\begin{itemize}
\item
$\hi_A$ is infinite dimensional and 
\item
$\hi_B$ is separable (with either $\dim(\hi_B)<\infty$ or $\dim(\hi_B)=\infty$); 
\end{itemize}
{\it for the rest of this section we assume that these conditions hold.} Now $\dim(\hi_A) \dim(\hi_B)-\dim(\hi_A)=\infty-\infty$ and $\dim(P^\perp(\hi_{AB}))$ may even vanish\footnote{Recall that $\N$ and $\N\times S$, $S\subseteq\N$, $S\ne\emptyset$, have the same cardinality, so that the basis $\{y_n\}_{n=1}^\infty$ of $Y(\hi_A)$ may also be the basis of $\hi_{AB}=\hi_A\otimes\hi_B$; see Eq.\,\eqref{jdfjhfdhj}.} if $Y(\hi_A)=\hi_{AB}$. 

\begin{example}\label{ex:rankinfinity}
If $\dim(\hi_B)>1$ then $\dim[\hi_A\otimes(\C\xi)]^\perp=\infty$. Thus, in this case $U$ and $\xi$ exist if and only if $\dim (P^\perp(\hi_{AB}))=\infty$, or equivalently $\text{rank} \ P^\perp = \infty$.
\end{example}

\begin{example}\label{ex:lifting}
Let $\hi_B:=\C$ so that $\dim(\hi_B)=1$ and $\hi_{AB} \cong \hi_A$. Define a linear isometry $Y: \hi_A \rightarrow \hi_{AB}$ via $Y=\sum_{n=1}^\infty\kb{h_{n+1}}{h_n}$. This prototype of an isometric operator, which is clearly not unitary, corresponds physically to a systematically leaking quantum channel $B \mapsto Y B Y^*$. Now $P^\perp = \id_{AB} - YY^* = |h_1\rangle \langle h_1 |$ so that $P^\perp( \hi_{AB}) = \C h_1$ and $\dim[\hi_A\otimes(\C\xi)]^\perp=0\ne 1=\dim(P^\perp( \hi_{AB}))$, where
$\xi\in\T:=\{\xi\in\C\,|\,|\xi|=1\}$. Thus, $U$ defined via Eq.\,(\ref{isom:extension}) does not extend to a unitary operator for any $\xi\in\mathbb T$. However, a classical result by Halmos \cite[Problem 222]{Halmos82} shows that any linear isometry $Y:\hi_A \rightarrow \hi_A$ can be dilated to a unitary mapping $U_+: \hi_A\oplus \hi_A \rightarrow \hi_A \oplus \hi_A$ by defining
\begin{eqnarray}
U_+:= \left( \begin{array}{lc}
Y & \id_A - Y Y^* \\
0 & - Y^*
\end{array} \right).
\end{eqnarray}
Now $U_+{\psi\choose0}= {Y \psi\choose0}$. Since $\hi_A \oplus \hi_A \cong \hi_A \otimes (\C \oplus \C)$, we note that growing the dimension of the auxiliary space $\hi_B$ just by one is enough to solve the unitary extension problem.
\end{example}

Motivated by the previous example, we end this section by showing that, if $\dim[\hi_A\otimes(\C\xi)]^\perp\ne \dim (P^\perp(\hi_{AB}))$, the unitary extension problem can be solved by adding one extra dimension of the auxiliary Hilbert space $\hi_B$. By doing this, we have
\begin{eqnarray}
\hi_A\otimes(\hi_B\oplus\C)\cong\hi_{AB}\oplus\hi_A\cong\hi_{AB}\times\hi_A
\end{eqnarray}
We then interpret ${Y:\,\hi_A\to\hi_{AB}}$ as an isometry ${Y_+:\hi_A \rightarrow \hi_{AB}\times\hi_A}$ via $Y_+\psi:=(Y\psi,0)$ and define $P_+ := Y_+ Y_+^* \in \mc P (\hi_{AB}\times \hi_A)$.
Now $P_+^\perp(\hi_{AB}\times \hi_A)$ and
$[\hi_A\otimes(\C\xi_+)]^\perp$, where $\xi_+\in\hi_B\oplus\C$ is a unit vector, are both separable infinite dimensional Hilbert spaces and are thus unitarily isomorphic. 


\section{Minimal normal measurement models: the discrete case}\label{sec:discrete}
In this section we study the implications of the unitary extension problem to quantum measurement theory. In particular, we prove that the size of minimal normal measurement realization of practically any quantum instrument equals to the sum of its Kraus ranks. However, there are some pathological cases in which this dimension has to be grown by one -- we completely isolate these cases to instruments having so-called rank-$\infty$ POVMs as their associated observables. Furthermore, we show that any $N$-outcome observable can be measured with an $N$-dimensional normal measurement.

\subsection{Instruments}
Let  $d$ and $N$ belong to $\{1,2,\ldots\}\cup\{\infty\}$. 
Suppose $\Mo=(\Mo_1,\,\Mo_2,\ldots)$ is a discrete $N$-outcome observable in a $d$-dimensional Hilbert space $\hi_d$ with an orthonormal basis $\{ h_n\}_{n=1}^d$. Let $\id_d$ be the identity operator of $\hi_d$.
Since each of the operators $\Mo_i$, $i<N+1$, is a positive non-zero bounded operator on $\hi_d$, we may write
\begin{eqnarray}\label{memo:decomposition}
\Mo_i=\sum_{k=1}^{m_i}\kb{d_{ik}}{d_{ik}},
\end{eqnarray}
where the vectors $d_{ik}\in\hi_d$, $k<m_i+1$, form a linearly independent set.
Also, there exist (linearly independent) vectors $g_{ik}\in\hi_d$, $k<m_i+1$, such that the following bi-orthogonality relation holds:
\begin{eqnarray}
\<d_{ik}|g_{i\ell}\>=\delta_{k\ell}.
\end{eqnarray}
For example, if $\Mo_i$ has a discrete spectrum (this always holds when $d<\infty$), then one can write
\begin{eqnarray}
\Mo_i=\sum_{k=1}^{m_i}\kb{d_{ik}}{d_{ik}}=\sum_{k=1}^{m_i}\lambda_{ik}\kb{\psi_{ik}}{\psi_{ik}},
\end{eqnarray}
where the eigenvectors $\psi_{ik}$ of $\Mo_i$ form an orthonormal set, the eigenvalues $\lambda_{ik}=\|d_{ik}\|^2$ are non-zero (and bounded by 1) and $d_{ik}:=\sqrt{\lambda_{ik}}\psi_{ik}$. Then we may define the vectors $g_{ik}$ via $g_{ik}:=\psi_{ik}/\sqrt{\lambda_{ik}}$ that satisfy the aforementioned properties.\footnote{
If $d=\infty$ and $\Mo_i$ has no discrete spectrum, we may use a trace-class operator
$T=\sum_{n=1}^\infty p_n\kb n n$, $p_n>0$, $\sum_{n=1}^\infty p_n<\infty$, to get a trace-class operator 
$T\Mo_iT$ with an eigendecomposition $T\Mo_i T=\sum_{k=1}^{m_i}\alpha_k\kb{\psi_k}{\psi_k}$.
Then, by defining $d_{ik}:=\sqrt{\alpha_k}T^{-1}\psi_k$ (and $g_{ik}:=T\psi_k/\sqrt{\alpha_k}$) we get Eq.\,(\ref{memo:decomposition}). } 

We say that $m_i$ is the multiplicity of the outcome $x_i\in\Omega_N$, or equivalently the {\it rank} of $\Mo_i=\Mo(\{x_i\})$. Note that $m_i\le d$ for all $i<N+1$. In particular, we call $\Mo$ \textit{rank-1}, if $m_i=1$ for all $i<N+1$. 
One particularly nice feature of rank-1 POVMs is that they are \textit{post-processing maximal} \cite{MaMu90, Buschemi2005, Pel2011} which physically means that they can be seen as ultimate fine-grainings of measurable quantities, after which no additional information can be extracted by means of performing additional measurements. 
As an opposite we define an observable $\Mo:\Sigma_N \rightarrow \mc L(\hi_\infty)$ to be rank-$\infty$ if $m_i = \infty$ for all $i<N+1$. A subclass of rank-$\infty$ observables is formed by trivial observables, $\Mo_i=p_i\id_\infty$, $p_i>0$, $\sum_{i=1}^N p_i=1$.

Let $\hM:\Sigma_N\to\li\big(\cal T(\hi_d)\big)$ be an instrument, such that $\Mo$ is its associated observable. Now
$\hM$ can be viewed as a collection $(\hM_1,\hM_2,\ldots)$ of CP trace non-increasing maps
$\hM_i:\,\hT(\hi_d)\to\hT(\hi_d)$, such that $\hM_i^*(\id_d)=\Mo_i$ for all $i<N+1$.
Each of $\hM_i^*$ has a minimal Kraus decomposition
\begin{eqnarray}\label{asdfasdf}
\hM_i^*(B)=\sum_{s=1}^{r_i}\Ao_{is}^*B\Ao_{is},\qquad B\in\li(\hi_d), 
\end{eqnarray}
where the Kraus operators $\Ao_{is}\in\li(\hi_d)$, $s<r_i+1$, are linearly independent and satisfy $\sum_s \Ao_{is}^* \Ao_{is} = \Mo_i$. Recalling Section \ref{sec:prel}, number $r_i$ is the Kraus rank of $\hM_i$ and we say that the instrument $\hM$ is \textit{rank-1} if $r_i=1$, for all $i<N+1$. Note that $r_i\le m_i d\le d^2$ \cite[Theorem 1]{Pel2013II}. 

Define the \textit{instrument's structure vectors} $\fii_{iks}:=\Ao_{is}\,g_{ik}\in\hi_d$ for which
\begin{eqnarray}
\sum_{s=1}^{r_i}\<\fii_{iks}|\fii_{i\ell s}\>=\<g_{ik}|\Mo_ig_{i\ell}\>=\delta_{k\ell}.
\end{eqnarray} 
Since $\Ao_{is}^* \Ao_{is} \le \Mo_i$
and
$\sum_{k=1}^{m_i} |g_{ik}\>\< d_{ik}| g_{i\ell} \> = g_{i\ell}$ 
we get
\begin{eqnarray}
\Ao_{is}=\sum_{k=1}^{m_i}\kb{\fii_{iks}}{d_{ik}}.
\end{eqnarray}
Especially, if $\hM$ is rank-1, \textit{i.e.}, $s=r_i=1$, we may write $\Ao_{i1}=\Ao_i$ and $\fii_{ik1}=\fii_{ik}$ to get
$\hM_i^*(B)=\Ao_{i}^*B\Ao_{i}$ and $\Ao_{i}=\sum_{k=1}^{m_i}\kb{\fii_{ik}}{d_{ik}}$ where $\<\fii_{ik}|\fii_{i\ell}\>=\delta_{k\ell}$.

Let then $\vec r:=(r_1,r_2,\ldots)$ be the \textit{Kraus rank vector} of $\hM$ and let $\hi_{\vec r}$ be a (separable) Hilbert space spanned by an orthonormal set 
$$
\{ e_{is}
 \ | \ i < N+1, \; s < r_i +1\}. 
$$ 
Obviously
$\dim\hi_{\vec r}=\sum_{i=1}^N r_i \leq N d^2$.
Let $\Po$ be a discrete sharp observable defined by 
\begin{eqnarray}
\Po_i=\sum_{s=1}^{r_i}\kb{e_{is}}{e_{is}}
, \qquad i<N+1.
\end{eqnarray}
Define a (bounded) linear map $Y:\,\hi_d\to\hi_d\otimes\hi_{\vec r}$ via
\begin{eqnarray}
Y\psi=\sum_{i=1}^N\sum_{s=1}^{r_i}\Ao_{is}\psi\otimes e_{is}, \qquad\psi\in\hi_d,
\end{eqnarray}
for which
\begin{eqnarray}
\<\psi|Y^*(B\otimes\Po_i)Y\fii \> &=&
\sum_{j,k=1}^N\sum_{s,t=1}^{r_i}\<\Ao_{js}\psi| B\Ao_{kt}\fii\>\stackrel{=\,\delta_{ji}\delta_{ki}\delta_{st}}{\<e_{js}| \Po_ie_{kt}\>} \nonumber \\
&=& \<\psi|\hM_i^*(B)\fii\>, \qquad \psi,\,\varphi \in \hi_d,
\end{eqnarray} implying that $\hM_i^*(B)=Y^*(B\otimes\Po_i)Y$. Now $Y^*Y=\sum_i\mc I_i^*(\id_d)=\id_d$, that is $Y$ is an isometry, and thus $\big(\hi_{\vec r},\Po,Y\big)$ constitutes a Stinespring dilation of $\hM$. 
It can be shown that the dilation is minimal, that is,
the closure of the span of vectors $(B\otimes\Po_i)Y\psi$, $B\in\li(\hi_d)$, $i<N+1$, $\psi\in\hi_d$, is the whole $\hi_d\otimes\hi_{\vec r}$ \cite{Pel2013I}. For instance, if $\mc I$ is rank-1, 
this follows immediately from 
$
(B\otimes\Po_i)Y g_{i1}=(B\Ao_{i} g_{i1})\otimes e_{i1}=h_n\otimes e_{i1}$,
when choosing $B=|h_n\rangle \langle \fii_{i1} | \in \mc L(\hi_d)$.

Fix then an arbitrary unit vector $\xi\in\hi_{\vec{r}}$ and define an operator
$U:\,\hi_d\otimes\C\xi\to\hi_d\otimes\hi_{\vec r}$ via
\begin{eqnarray}\label{memo:unitary}
U(\psi\otimes\xi)&:=&Y\psi=\sum_{i=1}^N\sum_{s=1}^{r_i}\Ao_{is}\psi\otimes e_{is} \\ \nonumber
&=&\sum_{i=1}^N\sum_{s=1}^{r_i}{\left(\sum_{k=1}^{m_i}\<d_{ik}|\psi\>\fii_{iks}\right)}\otimes e_{is}, \qquad\psi\in\hi_d.
\end{eqnarray}
The crucial question now is {\it ``when can $U$ be extended to a unitary operator of $\hi_d\otimes\hi_{\vec r}$?''}
If this unitary extension problem can be solved, one gets 
\begin{eqnarray}
\hM_i^*(B)=Y^*(B\otimes\Po_i)Y=Y_\xi^*U^*(B\otimes\Po_i)U Y_\xi,
\end{eqnarray}
where $Y_\xi:\,\hi_d\to\hi_d\otimes\hi_{\vec r}$ is an isometry, and thus one gets a  minimal\footnote{Since any normal measurement $(\ki, \Po', U', \xi')$ realizing $\mc I$ satisfies $\dim \ki \geq \dim \hi_{\vec{r}} = \sum_{i=1}^N r_i$ where $r_i$ is the Kraus rank of $\hM_i$, $i<N+1$ \cite[Theorem 3]{Pel2013II}.
} 
 normal measurement model $(\hi_{\vec r},\Po,U, \xi)$ realizing $\hM$; see Eq.\,(\ref{obs:dilation}).
 The analysis done in the previous section shows that this will always happen, if $d<\infty$. 
We also know that, if $d=\infty$, the extension succeeds, when the dimension of the auxiliary space $\hi_{\vec{r}}$ is increased by one. 
The following theorem shows that usually this extension is not necessary.

\begin{theorem}
Let $\hi$ be a separable Hilbert space and $\hM:\,\Sigma_N\to\li\big(\th\big)$ an instrument having a discrete $N$-outcome observable $\Mo=(\Mo_1,\Mo_2,\ldots)$ associated to it. Let $\big(\ki,\Po,Y\big)$ be a minimal Stinespring dilation of $\hM$. 
If for some $i<N+1$ the effect $\Mo_i$ has a finite rank, \textit{i.e.}, $0<m_i<\infty$, then for any unit vector $\xi\in\ki$ there exists a unitary operator 
$U$ on $\hi\otimes\ki$, such that $U(\psi\otimes\xi)=Y\psi$ for all $\psi\in\hi$.
\end{theorem}

\begin{proof}
We may assume that $\hi=\hi_\infty$. Furthermore, since a minimal Stinespring dilation is unique up to a unitary isomorphism, we may assume that $\ki=\hi_{\vec r}$ and
\begin{eqnarray}
Y\psi
=\sum_{i=1}^N\sum_{s=1}^{r_i}\Ao_{is}\psi\otimes e_{is}
=\sum_{i=1}^N\sum_{s=1}^{r_i}\sum_{k=1}^{m_i}\<d_{ik}|\psi\>\fii_{iks}\otimes e_{is}
\end{eqnarray}
Fix any $i<N+1$ satisfying $m_i<\infty$ and fix some $s<r_i +1$. Since $\dim \hi =\infty$, there exists an infinite number of linearly independent vectors orthogonal to the vectors $\varphi_{iks}$, $k\in\{1,2,\dots,m_i\}$. Let $\eta$ be one of these vectors, so that $\eta \perp \varphi_{iks}$ for all $k\le m_i$. Then $\eta \otimes e_{is} \perp \varphi_{i'k's'} \otimes e_{i's'},$ for all $i'<N+1$, $k'<m_{i'} +1$ and $s' < r_{i'}+1$, implying that 
\begin{equation}\label{lkjsdhfkljsdfh}
YY^*(\eta \otimes e_{is}) = 0. 
\end{equation}
Since the number of linearly independent vectors  $\eta$ satisfying \eqref{lkjsdhfkljsdfh} is infinite, $\id_{\hi\otimes \hi_{\vec r}} - YY^*$ is of infinite rank and the claim follows from Example \ref{ex:rankinfinity}. (Note that if $N=1$ then $\Mo_1=\id_\infty$ with $m_1=\infty$.)
\end{proof}

We have seen that in the discrete case the unitary extension of $U$ may fail only if the instrument's associated observable is rank-$\infty$, but even then the condition is not sufficient. In physical applications this is not a serious problem, since we may add one dimension to $\hi_{\vec r}$, \textit{i.e.}, define $\ki:=\hi_{\vec r}\oplus\C\cong\hi_{\vec r}\times\C$. 
{This addition of extra dimension allows us to interpret the system-apparatus --composite as a closed quantum system where the measurement dynamics is governed by unitary evolution.
Namely,}
by defining $Y_+\psi:=(Y\psi,0)$ for all $\psi\in\hi$ 
we see that $U_+(\psi\otimes\xi_+):=Y_+\psi$, where $\xi_+\in\ki$ is a unit vector, extends to a unitary operator $U_+$ on $\hi \otimes \ki$; see the end of Sect.\,\ref{sec:isometry}.
{One possibility to extend the pointer observable $\Po$ on $\ki$, is to add an extra outcome to $\Omega_N=\{x_1,\,x_2,\ldots\}$,} 
{\it i.e.}, define
$\Omega^+_{N}:=\Omega_N\cup\{x_0\}$ and $\Sigma_N^+:=2^{\Omega^+_N}$ where $x_0\notin\Omega_N$.
Now, for all $X\in\Sigma_N$ and $(\xi,c)\in\ki$, the equations $\Po_+(X)(\xi,c):=(\Po(X)\xi,0)$ and $\Po_+(\{\omega_0\})(\xi,c):=(0,c)$ define a new sharp pointer observable
$\Po_+:\,\Sigma^+_N\to\mc L(\ki)$. The triple $(\ki, \Po_+, Y_+)$ is a  Stinespring dilation of an instrument\footnote{
That is, $\hM_+^*(X_+)(B)=Y^*_+\big(B\otimes\Po_+(X_+)\big)Y_+$, $B\in\lh$, $X_+\in\Sigma_N^+$. Especially, $\hM_+(\{\omega_0\})=0$.} $\hM_+$
whose associated observable $\Mo_+$ is a (trivial) extension of $\Mo$. Indeed, $\Mo_+(\{\omega_i\})=\Mo_i$ for all $1\le i < N+1$ and $\Mo_+(\{\omega_0\})=0$. {Note that, the extra outcome $\omega_0$ is never obtained in the measurement since the probability of $\omega_0$ in any input state of the system is zero}.

We summarize the above observations in the following corollary.

\begin{corollary}\label{cor:instru}
Let $\hi$ be a separable Hilbert space and $\hM:\,\Sigma_N\to\li\big(\th\big)$ an instrument having a discrete $N$-outcome observable $\Mo=(\Mo_1,\Mo_2,\ldots)$ associated to it. Suppose that $\big(\hi_{\vec r}, \Po, Y)$ is a minimal Stinespring dilation of $\mc I$.
\begin{itemize}
\item[(i)] If $\Mo$ is \textit{not} rank-$\infty$ observable then $(\hi_{\vec{r}}, \Po, U, \xi )$ is a minimal normal measurement realization of $\hM$. Here the unitary operator $U$ on $\hi\otimes \hi_{\vec r}$ and unit vector $\xi\in\hi_{\vec r}$ satisfy $U(\psi \otimes \xi) = Y\psi$ for all $\psi \in \hi$.
\item[(ii)] Let $\Mo$ be rank-$\infty$. 
\begin{itemize}
\item[(a)] If there exists a unitary operator $U$ on $\hi\otimes \hi_{\vec r}$ and unit vector $\xi\in\hi_{\vec r}$, such that $U(\psi \otimes \xi) = Y\psi$ for all $\psi \in \hi$, then $(\hi_{\vec{r}}, \Po, U, \xi )$ is a minimal normal measurement realization of $\hM$.  
\item[(b)] Otherwise, in the above notions, $\big(\hi_{\vec{r}}\oplus\C, \Po_+, U_+, \xi_+ \big)$ is a normal measurement realization of the extended instrument ${\mc I}_+$ for which $\hM_+(X)=\hM(X)$ for all $X\subseteq\Omega_N$ and $\hM_+(\{\omega_0\})=0$.
\end{itemize}
\end{itemize}
\end{corollary}

\begin{remark}\label{unique} It has been proven in \cite{Pel2013II}, that the minimal normal measurement realizations are unique up to obvious unitary transformations. Furthermore, minimal normal measurement model can be isometrically embedded in any normal measurement model of the same device; see Theorem 3 in Ref.\,\cite{Pel2013II}. Therefore, we interpret minimal normal measurement models as effective descriptions of normal measurement protocols that contain only  essential measurement degrees of freedom.   
\end{remark}

\begin{example}\label{ex:rankinfinity2}
In the case of $N=1$, the instrument $\hM$ is a quantum channel and its associated observable is simply $\Mo_1=\id_d$, with the rank $m_1=d$.
Now one can choose $d_{1k}=g_{1k}=h_k$ and Eq.\,\eqref{asdfasdf} is a minimal Kraus decomposition of the channel $\hM_1$:
\begin{eqnarray}
\hM_1^*(B)=\sum_{s=1}^{r_1}\Ao_{1s}^*B\Ao_{1s},\qquad B\in\li(\hi_d), 
\end{eqnarray}
where $\Ao_{1s}=\sum_{k=1}^{d}\kb{\fii_{1ks}}{h_{k}}$ and $\sum_{s=1}^{r_1}\<\fii_{1ks}|\fii_{1\ell s}\>=\delta_{k\ell}$.
Moreover, $U(h_k\otimes\xi)=Y h_k=
\sum_{s=1}^{r_1}\fii_{1ks}\otimes e_{1s}$.
If $m_1=d=\infty$, it may happen that $U$ does not extend to a unitary operator. Especially, in the case $r_1=1$ one sees this clearly. Namely, then $U(h_k\otimes\xi)=\fii_{1k1}\otimes e_{11}$, where $\xi=te_{11}$, $t\in\mathbb T$,
so that $U$ is unitary if and only if the orthonormal set $\{\fii_{1k1}\}_{k=1}^\infty$ is a basis of $\hi_\infty$; compare to Example \ref{ex:lifting}.
In this case, $\Ao_{11}$ is unitary and $\hM_1$ is a unitary channel.
\end{example}

Together with Corollary \ref{cor:instru}, the above example has implications, for instance, in the field of open quantum systems. Namely, given a quantum channel $\mc E:\th\rightarrow\th$ these allow one to solve a minimal environment $\ki$ and composite unitary channel, from which $\mc E$ can be reduced, so that
\begin{eqnarray}\label{eq:channel}
\mc E(T) = \ptr{U \big(T\otimes \kb\xi\xi\big) U^*}, \qquad T \in \th.
\end{eqnarray}
Now $\dim\ki=\text{rank}(\mc E)$ if $\dim\hi<\infty$ and $\dim\ki$ is either $\text{rank}(\mc E)$ or $\text{rank}(\mc E)+1$ if $\dim\hi=\infty$
as described in Corollary \ref{cor:instru}. 
In the view of Remark \ref{unique}, $\ki$ may be interpreted as an \textit{effective environment} which can be isometrically embedded into any other environment inducing $\mc E$ in the above form \eqref{eq:channel}. 


\subsection{Observables}

We will next examine normal measurement models $(\ki, \Po, U, \xi)$ of a discrete $N$-outcome observable $\Mo$ in a separable Hilbert space. As shown in Corollary \ref{cor:instru}, the size of the apparatus is related to the number of outcomes $N$ of the observable and the Kraus ranks $r_i$ corresponding to the inducing instrument via $\dim \ki \geq \sum_{i=1}^N r_i$. Since we are interested in finding a minimal normal measurement, we may assume that the instrument is rank-1. We denote $e_{i1}=e_i$ and $\vec1=(1,1,1,\ldots)$ (with $N$ elements) so that $\dim\hi_{\vec 1}=N$. In this case Eq.\,(\ref{memo:unitary}) simplifies to 
\begin{eqnarray}\label{memo:rank1unitary}
U(\psi\otimes\xi)=Y\psi=\sum_{i=1}^N\Ao_{i}\psi\otimes e_{i}
=\sum_{i=1}^N\sum_{k=1}^{m_i}\<d_{ik}|\psi\>\fii_{ik}\otimes e_{i},
\end{eqnarray}
for all $\psi\in\hi_d$, where the instrument's structure vectors $\fii_{ik}$ satisfy the orthogonality relation $\< \fii_{ik}|\fii_{i\ell}\>=\delta_{k\ell}$.
As we have seen, $U$ always extends to a unitary operator, when $d<\infty$.
The following example illustrates, that in the case of $d=\infty$ (and $m_i\equiv\infty$), the extension problem is related to the choice of the vectors $\fii_{ik}$.

\begin{example}\label{ex:structurevectors}
(i) Let $\Mo: \Sigma_N \rightarrow \mc L(\hi_d)$ be an $N$-outcome {\it sharp} observable. Note that because $\Mo$ is sharp, $\{d_{ik}\}_{i,k}$ forms an orthonormal basis of $\hi_d$. By choosing $\fii_{ik}=d_{ik}$ the operator $U$ defined in Eq.\,(\ref{memo:rank1unitary}) extends to a unitary operator and one gets minimal, so called \textit{von Neumann--L\"uders model} for $\Mo$ \cite{QTM96}. Indeed, now $U(\psi \otimes \xi)=\sum_{i=1}^N\sum_{k=1}^{m_i}\<d_{ik}|\psi\>d_{ik}\otimes e_{i} = Y \psi$ for all $\psi\in\hi_d$ showing that $Y= \sum_{i=1}^N \sum_{k=1}^{m_i}  |d_{ik} \otimes e_i\> \<d_{ik}|$. Then $P=YY^*=\sum_{i=1}^N \sum_{k=1}^{m_i}  |d_{ik}\>\<d_{ik}| \otimes |e_i\>\<e_i| = \sum_{i=1}^N \Mo_i \otimes P[e_i]$ implying that 
\begin{eqnarray} 
P^\perp &=& \id_{\hi_d \otimes \hi_{\vec{1}}} - P = \sum_{i=1}^N \Mo_i \otimes (\id_{\hi_{\vec{1}}}-P[e_i]) \nonumber\\
&=&  \sum_{i=1}^N \Mo_i \otimes P[e_i]^\perp.  
\end{eqnarray} 
This with the property $\Mo_i\Mo_j=\delta_{ij}\Mo_i$ further imposes that 
$\rnk{P^\perp} = \sum_{i=1}^N \rnk{\Mo_i} \cdot \rnk{P[e_i]^\perp} =  \sum_{i=1}^N m_i (N-1) = d(N-1)$.
By Examples \ref{ex:rankinfinity} and \ref{ex:rankinfinity2}, we know that this always implies that $U$ extends to unitary on $\hi_d \otimes \hi_{\vec{1}}$.

(ii) Let $\Mo$ be a rank-$\infty$ sharp observable with $N>1$ and let $\{h_n\}_{n=1}^\infty$ be an orthonormal basis of $\hi_\infty$. If we choose $\fii_{ik}=h_k$ then
$$Y = \sum_{i=1}^N \sum_{k=1}^{m_i} |h_{k}\otimes e_{i}\rangle \langle d_{ik}|$$ and we see that $U$ does not extend to a unitary operator. Indeed, now $P=YY^*=  \sum_{i,k} |h_{k}\otimes e_{i}\rangle \langle h_{k}\otimes e_{i}| = \id_{\hi_{\infty} \otimes \hi_{\vec{1}}}$ implying that $P^\perp=0$ and the claim follows from Example \ref{ex:rankinfinity}.
\end{example}


\begin{example}\
Let $d=\infty$ and $\Mo: \Sigma_2 \rightarrow \mc L(\hi_\infty)$ a $2$-outcome observable such that $\Mo_1$ has a discrete spectrum.
Then there exists an orthonormal basis $\{h_n\}_{n=1}^\infty$ such that
$\Mo$ is of the form $\Mo_1 = \sum_{k=1}^\infty \lambda_{1k} |h_k\rangle \langle h_k|$ and $\Mo_2 = \id_{\infty} - \Mo_1 = \sum_{k=1}^\infty \lambda_{2k} |h_k \rangle \langle h_k|$ with $\lambda_{2k}=1-\lambda_{1k}$, $\lambda_{1k} \in [0,1]$ for all $k$. Note that rank of $\Mo_i$ is the number of non-zero numbers $\lambda_{ik}$, $k\ge1$.
We are interested in the case where $m_1=m_2=\infty$. For example, this holds when
there exists an infinite number of `eigenvalues' $0<\lambda_{1k}<1$ of $\Mo_1$ which is assumed next.
Note that, if $\Mo_1$ is a projection, then $\Mo_2=\Mo_1^\perp$ and each $\lambda_{1k}$ (and thus $\lambda_{2k}$) is either 0 or 1.

Define $\hi_{\vec{1}}=\C e_1+\C e_2\cong \C^2$ as above and fix some unit vector $\xi \in \hi_{\vec{1}}$. Furthermore, define an operator $U: \hi_\infty \otimes \C \xi \rightarrow \hi_\infty \otimes \hi_{\vec{1}}$ as in Eq.\,(\ref{memo:rank1unitary}), that is, 
\begin{eqnarray}
U(\psi \otimes \xi) &=& \sum_{i=1}^2 \sum_{k=1}^\infty \sqrt{\lambda_{ik}} \ip{h_k}{\psi} \varphi_{ik}\otimes e_i \nonumber \\
&=&  \sum_{k=1}^\infty  \ip{h_k}{\psi} \left(\sqrt{\lambda_{1k}} \varphi_{1k}\otimes e_1+ \sqrt{1-\lambda_{1k}} \varphi_{2k}\otimes e_2\right) \nonumber\\
&=& Y \psi, \qquad \psi \in \hi_\infty,
\end{eqnarray}
where the vectors $\fii_{ik}\in\hi_\infty$ are such that $\< \fii_{ik}|\fii_{i\ell}\>=\delta_{k\ell}$ for all $k,\ell\in\{1,2,\dots\}$ satisfying $\lambda_{ik}\lambda_{i\ell}\ne 0$.
Due to the factors $\sqrt{\lambda_{ik}}$ above, without restricting generality, we set $\fii_{ik}:=0$ if $\lambda_{ik}=0$.
Denote  $\Psi_{ik}:= \varphi_{ik}\otimes e_i$. Now
$Y = \sum_{k=1}^\infty  |\sum_{i=1}^2 \sqrt{\lambda_{ik}} \Psi_{ik}\rangle\langle h_k|$ and therefore $YY^* = \sum_{k=1}^\infty\sum_{i,j=1}^2  |\sqrt{\lambda_{ik}} \Psi_{ik}\rangle\langle \sqrt{\lambda_{jk}} \Psi_{jk}|$.
Since for all $k,\,\ell=1,2,\ldots$, 
$
\ip{ \sum_{j=1}^2 \sqrt{\lambda_{jk}} \Psi_{jk} \, } {\sqrt{\lambda_{2\ell}}\Psi_{1\ell}-\sqrt{\lambda_{1\ell}}\Psi_{2\ell}}=0
$, then also $ YY^*(\sqrt{\lambda_{2\ell}} \Psi_{1\ell}-\sqrt{\lambda_{1\ell}} \Psi_{2\ell})=0 $ and
we conclude that the rank of $P^\perp:=\id_{\infty} - Y Y^*$ is infinite. Therefore by Example \ref{ex:rankinfinity}, $U$ can be extended to a unitary operator $U$ on $\hi_\infty \otimes \hi_{\vec{1}}$ regardless of the choice of vectors $\fii_{ik}$. We immediately observe that this condition is due to infinite amount of numbers  $0<\lambda_{1k}<1$ of $\Mo_1$. In general this is not a necessary condition for $\Mo$ and  Example \ref{ex:structurevectors} (ii) shows that the extendability of $U$ strongly depends on the choice of vectors $\fii_{ik}$.
\end{example}

To end this section, we construct a minimal normal measurement model of an arbitrary rank-$\infty$ discrete $N$-outcome observable $\Mo$. Let us fix a basis $\{h_k\}_{k=1}^\infty$ of $\hi_\infty$ and choose the vectors $\fii_{ik}:=h_{2k}$ for all  $i<N+1$ and $k\in\{1,2,\ldots\}$. Obviously these vectors may be used to define an isometry $Y$ via Eq.\,(\ref{memo:rank1unitary}) and a rank-1 instrument, that has $\Mo$ as the associate observable.
Now $YY^* (h_{2k+1} \otimes e_i) =0$ for all $k$ proving that $\id_{\hi_\infty \otimes \hi_{\vec{1}}} - YY^*$ has infinite rank and therefore the operator $U$ defined in Eq.\,\eqref{memo:rank1unitary} has an unitary extension; see Example \ref{ex:rankinfinity}. We summarize the observations of this subsection in the following proposition.
\begin{proposition}\label{prop:observables}
Let $\Mo:\Sigma_N \rightarrow \lh$ be a discrete $N$-outcome observable. There exists a normal $\Mo$-measurement $( \hi_{\vec{1}}, \Po, U, \xi )$ such that $\dim \hi_{\vec{1}}=N$. {Furthermore, the apparatus' Hilbert space $\ki$ of any normal $\Mo$-measurement satisfies $\dim \ki \geq N$.}
\end{proposition}

In connection to the above result, a complete characterization of the minimal normal measurements of discrete $N$-outcome sharp observables has been given in \cite{BelCasLah90}. Furthermore, it has been proven in \cite[Prop.\,10]{HeiTuk2014} that, if $\Mo$ is a discrete $N$-outcome sharp observable, then $N$ is actually the minimal dimension of any, \textit{i.e.}, not necessarily normal, $\Mo$-measurement. In the case of non-sharp observables one cannot draw the same conclusion. 
For instance, any trivial observable $\Mo(\{ x_i\}) = p(\{ x_i\}) \, \id_\hi$, where $p: \Sigma_N \rightarrow [0,1]$ is a discrete probability measure, can be realized with one-dimensional non-normal measurement model $(\C, p, \id_{\hi \otimes\C}, \xi)$, $\xi \in \mathbb T$, whereas each minimal normal measurement model is $N$-dimensional.

\section{Minimal normal measurement models: the continuous case}\label{sec:continuous}
In this  section, we focus on normal measurements of an arbitrary observable and briefly study the implications of the unitary extension problem in this context. Throughout this section, we assume that $\hi$ is a separable Hilbert space and we let $\hM : \Sigma\to\li\big(\th\big)$ be an instrument and $\Mo: \Sigma \rightarrow \lh$ be its associated observable.

Let $(\hd', \Po', Y)$ (resp.\ $(\hd, \Po, J)$) be a minimal Stinespring dilation for $\hM$ (resp.\ a minimal Naimark dilation for $\Mo$),
where $\hd'=\int_\Omega^\oplus\hi'_{r(x)}\d\mu(x)$ and $\hd=\int_\Omega^\oplus\hi_{m(x)}\d\mu(x)$
are \textit{direct integral} Hilbert spaces and $\mu:\,\Sigma\to[0,1]$ is a probability measure\footnote{More generally, $\mu$ can be chosen to be any $\sigma$-finite measure,  such that $\mu$ and $\Mo$ are mutually absolutely continuous.} such that $\mu$ and $\Mo$ are mutually absolutely continuous \cite{Pel2013II}. For instance, one may choose $\mu$ as a mapping $X\mapsto \tr{\Mo(X) \, \rho}$ where $\rho$ is a \textit{faithful state}, \textit{i.e.}, does not have eigenvalue $0$.
Moreover, $\Po$ and $\Po'$ are the \textit{canonical spectral measures} of the corresponding direct integral Hilbert spaces, \textit{i.e.}, of the form $X \mapsto \hat{\chi}_X$ (the characteristic function of $X$).
The dimension $m(x)\le\dim\hi$ of the \textit{fiber} Hilbert space $\hi_{m(x)}$ is called the multiplicity of the outcome $x\in\Omega$ and 
the dimension $r(x)$ of $\hi'_{r(x)}$ is the pointwise Kraus rank of the instrument $\hM$ \cite{Pel2013I}.
We say that $\Mo$ is {\it rank-$\infty$} if $m(x)=\infty$ for $\mu$-almost all $x\in\Omega$ (and thus $\dim\hi=\infty$).

Since $\hM^*(X)(B) = Y^* \left(  B \otimes \Po(X)\right) Y$ for all $X\in\Sigma$ and $B\in\lh$, one gets
\begin{eqnarray}\label{cont:isom}
\Mo(X) = J^* \Po(X) J=Y^* \left(  \id_\hi \otimes \Po'(X)\right) Y, \qquad X \in \Sigma
\end{eqnarray}
and there exists an unique decomposable isometry $C:\,\hd\to\hi\otimes\hd'$,
\begin{eqnarray}\label{cont:decomp}
C=\int_\Omega^\oplus C_x\d\mu(x),
\end{eqnarray}
such that $CJ=Y$, where $C_x:\,\hi_{m(x)}\to\hi\otimes\hi'_{r(x)}$ is an isometry for all $x\in\Omega$  \cite[Theorem 1]{Pel2013II}.
Now the question is does 
\begin{eqnarray}\label{cont:unitary}
U(\psi\otimes\xi)=Y\psi=CJ\psi, \qquad \psi\in\hi,
\end{eqnarray}
extend to a unitary operator on $\hi\otimes\hd'$. As we have seen before, the choice of the unit vector $\xi\in\hd'$ is irrelevant.
Moreover, the extension is always possible, if $\dim\hi<\infty$ or $\hd'$ is not separable. Recall that $\hd'$ is separable if and only if the measure space $(\Omega,\Sigma,\mu)$ has a countable basis \cite[Remark 2]{Pel2013I}. This always holds when the measurable space $(\Omega,\Sigma)$ is countably generated, \textit{e.g.}, $\Sigma$ is the Borel $\sigma$-algebra of a locally compact second countable Hausdorff space $\Omega$.

\begin{theorem}\label{thm:cont}
Let $\hi$ be a separable Hilbert space and $\mc I: \Sigma \rightarrow \mc L (\th)$ an instrument with the associate observable $\Mo:\Sigma\rightarrow\lh$, and let $\big(\hd', \Po', Y\big)$ be a minimal Stinespring dilation of $\mc I$. If $\Mo$ is \textit{not} rank-$\infty$, then for any unit vector $\xi\in\hd'$ there exists a unitary $U$ operator on $\hi \otimes \hi_{\oplus}'$, such that $U(\psi \otimes \xi) = Y \psi$ for all $\psi \in \hi$. 
\end{theorem}

\begin{proof}
We assume, without loss of generality, that $\dim \hi = \infty$. 
If $\Mo$ is not rank-$\infty$, then there exists a positive integer $m<\infty$, such that the $\mu$-measurable set $\Omega_m:=\{x\in\Omega|\,m(x)=m\}$ has a (finite) positive measure, $\mu(\Omega_m)>0$. 
Fix a minimal Naimark dilation $(\hd, \Po, J)$ of $\Mo$ so that Eqs.\,\eqref{cont:isom} and \eqref{cont:decomp} hold. By Examples \ref{ex:rankinfinity} and \ref{ex:rankinfinity2} one needs to show that ${\rm rank}\, P^\perp=\infty$ where
\begin{eqnarray}
P:=YY^*=CJJ^*C^*\le CC^*=:Q.
\end{eqnarray}
Hence ${Q^\perp \leq P^\perp}$ and it is enough to show that ${\rm rank}\, Q^\perp=\infty$.
Now $Q=\int_\Omega^\oplus C_xC_x^*\d\mu(x)$ is a decomposable projection acting in $\hi\otimes\hd'\cong\int_\Omega^\oplus(\hi\otimes\hi'_{r(x)})$ where
the fibers $\hi\otimes\hi'_{r(x)}$ are infinite dimensional for $\mu$-almost all $x\in\Omega$.
Note that, for any $x\in\Omega_m$, the image $C_x\left(\hi_{m(x)}\right)$ of the isometry $C_x$ is an $m$-dimensional subspace of $\hi\otimes\hi'_{r(x)}$ which implies that 
the rank of the projection $C_xC_x^*$ is $m<\infty$.
Choose then a (weakly) $\mu$-measurable field of orthonormal bases $\{e_k(x) |\, x\in\Omega\}_{k=1}^\infty$ of $\hi\otimes\hd'$ such that
$C_xC_x^*=\sum_{k=1}^m\kb{e_k(x)}{e_k(x)}$ for all $x\in\Omega_m$ and define orthonormal  vectors $\eta_\ell\in\hi\otimes\hd'$, $\ell\in\{1,2,\ldots\}$, by
\begin{eqnarray}
\eta_\ell(x)=
\begin{cases}
\mu(\Omega_m)^{-1}e_{m+\ell}(x) , & x\in\Omega_m \\
0, & x\in\Omega\setminus\Omega_m .\\
\end{cases}
\end{eqnarray}
Since for all $\ell\in\{1,2,\ldots\}$ and $x\in\Omega$ one gets $(Q\eta_\ell)(x)=C_xC_x^*\eta_\ell(x)=0$, that is, $Q^\perp\eta_\ell=\eta_\ell$, we conclude that 
${\rm rank}\, Q^\perp=\infty$ and  $U$ extends to a unitary operator.
\end{proof}

Let $\Mo$ be rank-$\infty$ (and hence $\dim\hi=\infty$). If $\hd'$ is separable, then it may happen that $U$ does not extend to unitary operator on $\hi \otimes \hi_{\oplus}'$; see Example \ref{ex:structurevectors}.
Analogously to before, this problem may be solved by adding one dimension to $\hd'$, \textit{i.e.},
by defining $\ki:=\hd'\oplus\C\cong\hd'\times\C$ and
$Y_+\psi:=(Y\psi,0)$ for all $\psi\in\hi$,
we see that 
$U_+(\psi\otimes\xi_+)=Y_+\psi$ extends to a unitary operator $U_+$ on $\hi \otimes \ki$ for any unit vector $\xi_+\in\ki$.
In addition, let $x_0$ be such that $x_0\notin\Omega$, define $\Omega_+:=\Omega\cup\{x_0\}$, and let $\Sigma_+$ be the smallest $\sigma$-algebra of $\Omega_+$ containing $\Sigma$. Now
$\Po'_+(X)(\xi,c):=(\Po'(X)\xi,0)$ and $\Po'_+(\{x_0\})(\xi,c):=(0,c)$, for all $X\in\Sigma$ and $(\xi,c)\in\ki$,
determines a PVM $\Po'_+:\,\Sigma_+\to\kh$ and we get a Stinespring dilation $(\ki, \Po'_+, Y_+)$ for an instrument $\hM_+$ such that $\hM_+(X)=\hM(X)$, $X\in\Sigma$, and $\hM_+(\{x_0\})=0$.
Moreover, the associate observable satisfies conditions $\Mo_+(X)=\Mo(X)$, $X\in\Sigma$, and $\Mo_+(\{x_0\})=0$.

We have the following corollary:
\begin{corollary}\label{cor:continuous} Let $\hi$ be a separable Hilbert space and $\mc I: \Sigma \rightarrow \mc L \big(\th\big)$ an instrument having an observable $\Mo:\Sigma\rightarrow\lh$ associated to it and let $\big(\hd', \Po', Y\big)$ be a minimal Stinespring dilation for $\mc I$.
\begin{itemize}
\item[(i)] If $\Mo$ is not rank-$\infty$, then $\big(\hd', \Po', U, \xi \big)$ is a minimal normal measurement realization of $\mc I$, where $U$ is a unitary operator on $\hi \otimes \hd'$ satisfying $U(\psi \otimes \xi) = Y\psi$, for all $\psi \in \hi$, and $\xi\in\hd'$ is a unit vector.  
\item[(ii)] Let $\Mo$ be rank-$\infty$.
\begin{itemize} \item[(a)]If there exists a unitary operator $U$ on $\hi \otimes \hd'$ and a unit vector $\xi\in\hd'$ such that $U(\psi \otimes \xi) = Y\psi$ for all $\psi \in \hi$, then $\big(\hd', \Po', U, \xi \big)$ is a minimal normal measurement realization of $\mc I$. 
\item[(b)]Otherwise, in the above notions, $\big(\hd'\oplus\C, \Po_+', U_+, \xi_+ \big)$ is a normal measurement realization of 
the extended instrument ${\mc I}_+$ for which $\hM_+(X)=\hM(X)$ for all $X\in \Sigma$ and $\hM_+(\{x_0\})=0$.
\end{itemize} 
\end{itemize}
\end{corollary}
Also, Proposition \ref{prop:observables} can be generalized for an arbitrary observable $\Mo$ by choosing the isometry $C$, \textit{i.e.}, the corresponding $\Mo$-compatible instrument $\hM$, in such a way that $U$ in Eq.\,\eqref{cont:unitary} extends unitarily. In particular, if $r(x)=1$ for $\mu$-almost all $x\in\Omega$, that is, $\hM$ is a rank-1 instrument, then one can choose the decomposing isometries
$C_x:\,\hi_{m(x)}\to\hi\otimes\hi'_{1}\cong\hi$ of Eq.\,\eqref{cont:decomp} so that the rank of the complement of $C_xC_x^*=\sum_{k=1}^{m(x)}\kb{e_k(x)}{e_k(x)}$ is infinite (even if $m(x)=\infty$). For example, pick an orthonormal basis $\{h_n\}_{n=1}^\infty$ of $\hi$ and define
$e_{k}(x):=h_{2k}$ for all $k<m(x)+1$. Note that $\hd'\cong L^2(\mu)$, so that any vector of $\hi\otimes \hd'\cong L^2(\mu)\otimes\hi$ can be interpreted as a $\mu$-square integrable function from $\Omega$ to $\hi$. Then the constant functions $x\mapsto h_n$ form a $\mu$-measurable field of orthonormal bases of $L^2(\mu)\otimes\hi$. 
Now $CC^*h_{2k+1}=0$ for all $k$ implying that $\id_{L^2(\mu)\otimes\hi} - CC^*$ has infinite rank and the extension succeeds; see Example \ref{ex:rankinfinity} and the proof of Theorem \ref{thm:cont}.

Finally, it has been proven in \cite[Remark 3]{Pel2013II} that the dimension of the Hilbert space $\ki$ of any normal $\Mo$-measurement satisfies $\dim\ki\ge \dim L^2(\mu)$. We end with the following proposition that summarizes these facts.
\begin{proposition}\label{prop:contobservables}
Let $\Mo:\Sigma \rightarrow \lh$ be an observable and $\mu:\Sigma \rightarrow [0,1]$  a probability measure such that $\mu$ and $\Mo$ are mutually absolutely continuous. There exists a normal $\Mo$-measurement $(L^2(\mu), X\mapsto\hat\chi_X, U, \xi )$. Furthermore, the apparatus' Hilbert space $\ki$ of any normal $\Mo$-measurement satisfies $\dim \ki \geq \dim L^2(\mu)$.
\end{proposition}
\newpage

\section*{Errata for articles \cite{Pel2013II, Pel2014, PeENT}} \label{errata}
In this Errata, we point out errors found in Refs.\,\cite{Pel2013II, Pel2014, PeENT}. The errors are due to the faulty assumption that an operator $U$ defined via $U(\psi \otimes \xi) = Y \psi$ for all $\psi \in \hi$, where $\xi\in\ki$ is some unit vector and $Y:\hi \rightarrow \hi \otimes \ki$ is an isometry, always extends to a unitary operator on $\hi \otimes \ki$.
The errors may occur only in the pathological case when the associate observable of an instrument is rank-$\infty$; see Theorem \ref{thm:cont}.

Errata for \cite{Pel2013II} (J.-P.\ Pellonp\"{a}\"{a}, {\it J.\ Phys.\ A: Math.\ Theor}.\ \textbf{46}, 025303, 2013):
\begin{itemize}
\item Page 10, Example 3: Let $\{h_n\}$ be ON-basis of $\hi$ and $\xi_\oplus$ a unit vector in $\hi_\oplus$. Since the sets $\{Y h_n\}_{n=1}^{\dim(\hi)}$ and $\{h_n \otimes \xi_{\oplus}\}_{n=1}^{\dim(\hi)}$ are ON-sets having the same cardinality, it was claimed that one can define  a unitary operator on $\hi \otimes \hi_\oplus$ by setting $U_{\xi_\oplus}(h_n \otimes \xi_\oplus) := Y h_n$. This claim is false in general; see Section \ref{sec:isometry} (Example \ref{ex:lifting}) of this paper. However, Theorem 3 of \cite{Pel2013II} is correct.
\item The inequality $\dim \hi_\oplus \geq \dim \hi$ (where $\hi_\oplus$ is the minimal Stinespring dilation space of an instrument and $\hi$ is the system's Hilbert space) in the second last line of page 10 is false. For a counter-example consider the Stinespring dilation of an instrument, with a 1-outcome associated observable, and assume that $\dim \hi \geq 2$.
\item Corollary 2 on page 11 is false in the case when $U_{\xi_\oplus}(\psi \otimes \xi_\oplus) = Y \psi$ does not extend to a unitary operator; see our Corollaries \ref{cor:instru} and \ref{cor:continuous} for comparison.
\item Remark 3 on page 11 is slightly misleading. In the case of a rank-$\infty$ POVM, one must choose the isometries $C_x$ properly in $[U(\psi \otimes \xi )](x) = (C_x \Ao(x) \psi) \otimes 1$ to define a unitary operator. In other cases, the isometries $C_x$ can be arbitrary.
\end{itemize}

Errata for \cite{Pel2014} (J.-P.\ Pellonp\"{a}\"{a}, {\it Found.\ Phys}.\ {\bf 44}, 71--90, 2014):
\begin{itemize}
\item The operators $U$ defined in Eq.\,(5) on page 80, in Example 6 on page 81, in Example 7 on page 82 and the sixth last line on page 88 may not in general extend to unitary operators. If the multiplicities $m_i$ or $m(x)$ are all infinite one must choose the instrument's structure vectors $\varphi_{iks}$, $\varphi_{ik}$ and $\varphi_k(x)$ properly; see Example \ref{ex:structurevectors}.
\end{itemize}

Errata for \cite{PeENT} (J.-P.\ Pellonp\"a\"a, {\it Phys.\ Lett.\ A} {\bf 376}, 3495--3498, 2012):
\begin{itemize}
\item The operator $U_{SA}$ defined below Eq.\,(3) on page 3497 may not extend to a unitary operator if all the multiplicities $m_i$ are infinite, {\it i.e.}, $m_i = \infty$ for all $i<N+1$; see Example \ref{ex:structurevectors}. One may add one extra dimension to the apparatus space to solve the problem; see Corollary \ref{cor:instru} and the construction above. This addition does not affect the subsequent calculations.
\end{itemize}


\section*{Acknowledgments} The authors wish to thank Dr.\ Teiko Heinosaari for useful discussions and comments on the manuscript.  
M.\ T.\ acknowledges the financial support from the University of Turku Graduate School (UTUGS).




\begin{thebibliography}{10}

\bibitem{QTM96}
P. Busch,  P.J. Lahti and P. Mittelstaedt.
\newblock {\it The quantum theory of measurement}.
\newblock Springer, 1996.
\newblock 2nd rev. ed.

\bibitem{Ozawa84}
M.~Ozawa.
\newblock Quantum measuring processes of continuous observables.
\newblock {\it J. Math. Phys.} {\bf 25}, 79--87, 1984.

\bibitem{AkhiGlaz93}
N. I. Akhiezer, I. M. Glazman.
\newblock {\it Theory of linear operators in Hilbert space. } 
\newblock Dover Publications, Inc., New York, 1993.
\newblock Engl. transl. Merlynd Nestell. Two volumes bound as one.

\bibitem{Halmos82}
P.R. Halmos.
\newblock {\it A Hilbert space problem book}.
\newblock Springer-Verlag, New York, 1974.

\bibitem{Pel2013II}
J.-P. Pellonp\"{a}\"{a}.
\newblock Quantum instruments: II. Measurement theory.
\newblock{\it J. Phys. A: Math. Theor.} {\bf 46}, 025303, 2013.

\bibitem{Pel2014}
J.-P. Pellonp\"{a}\"{a}.
\newblock Complete Measurements of Quantum Observables.
\newblock {\it Found. Phys.} {\bf 44}, 71-90, 2014.

\bibitem{PeENT} 
J.-P.\ Pellonp\"a\"a.
\newblock Complete quantum measurements break entanglement. 
\newblock {\it Phys. Lett. A} {\bf 376}, 3495-3498, 2012.

\bibitem{LaYl04}
P. Lahti and K. Ylinen.
\newblock Dilations of positive operator measures and bimeasures related to quantum mechanics.
\newblock {\it Math. Slovaca} {\bf 54}, 169--189, 2004.


\bibitem{MaMu90}
H. Martens and W. M. de Muynck.
\newblock Nonideal Quantum Measurements.
\newblock {\it Found. Phys.} {\bf 20}, 255-281, 1990.

\bibitem{Buschemi2005}
F. Buschemi and M. Keyl and G.M. D'Ariano and P. Perinotti and R. Werner.
\newblock Clean positive operator valued measures.
\newblock {\it J. Math. Phys.} {\bf 46}, 082109, 2005.

\bibitem{Pel2011}
J.-P. Pellonp\"{a}\"{a}.
\newblock On coexistence and joint measurability of rank-1 quantum observables.
\newblock {\it J. Phys. A: Math. Theor.} {\bf 47}, 052002, 2014.

\bibitem{Pel2013I}
J.-P. Pellonp\"{a}\"{a}.
\newblock Quantum instruments: I. Extreme instruments.
\newblock {\it J. Phys. A: Math. Theor.} {\bf 46}, 025302, 2013.

\bibitem{BelCasLah90}
E.G. Beltrametti and G. Cassinelli and P. Lahti.
\newblock Unitary measurements of discrete quantities in quantum mechanics.
\newblock {\it J. Math. Phys.} {\bf 31}, 91-98, 1990.

\bibitem{HeiTuk2014}
T. Heinosaari and M. Tukiainen.
\newblock Notes on deterministic programming of quantum observables and channels.
\newblock {\it Quantum Inf. Process.} {\bf 14}, 3079-3114, 2015.

\end{thebibliography}
\end{document}